\documentclass[11pt,letterpaper]{article}
\usepackage{color,ifpdf,latexsym,graphicx,url}
\usepackage[margin=1in]{geometry}
\usepackage{amsmath,amssymb,amsthm}
\usepackage{gensymb}
\usepackage{xcolor}
\usepackage{colortbl}
\usepackage{multirow}
\usepackage{enumitem}
\usepackage{subcaption}
\usepackage{overpic}

\usepackage{hyperref}
\hypersetup{breaklinks,bookmarks,bookmarksnumbered,bookmarksopen,bookmarksopenlevel=2}
{\makeatletter \hypersetup{pdftitle={\@title}}}

{\makeatletter
 \gdef\xxxmark{%
   \expandafter\ifx\csname @mpargs\endcsname\relax 
     \expandafter\ifx\csname @captype\endcsname\relax 
       \marginpar{xxx}
     \else
       xxx 
     \fi
   \else
     xxx 
   \fi}
 \gdef\xxx{\@ifnextchar[\xxx@lab\xxx@nolab}
 \long\gdef\xxx@lab[#1]#2{\textbf{[\xxxmark #2 ---{\sc #1}]}}
 \long\gdef\xxx@nolab#1{\textbf{[\xxxmark #1]}}
}

{\makeatletter \gdef\fps@figure{!htbp}}

\def\makecell#1{{\def\arraystretch{1}\begin{tabular}{@{}c@{}}#1\end{tabular}}}

\let\realbfseries=\bfseries
\def\bfseries{\realbfseries\boldmath}

\newtheorem{theorem}{Theorem}
\newtheorem{lemma}[theorem]{Lemma}
\newtheorem{corollary}[theorem]{Corollary}
\newtheorem{definition}[theorem]{Definition}
\newtheorem{question}{Open Problem}

\newcommand{\Celeste}{\textsc{Celeste}}
\newcommand{\ZeroCeleste}{\textsc{Zero-Player Celeste}}
\let\epsilon=\varepsilon
\def\defn#1{\textbf{\textit{\boldmath #1}}}

\title{Celeste is PSPACE-hard}
\author{%
  Lily Chung%
    \thanks{MIT Computer Science and Artificial Intelligence Laboratory,
    32 Vassar St., Cambridge, MA 02139, USA,
    \protect\url{{lkdc,edemaine}@mit.edu}}
\and
  Erik D. Demaine\footnotemark[1]
}
\date{}

\begin{document}
\maketitle

\begin{abstract}
We investigate the complexity of the platform video game Celeste.  We prove that navigating Celeste is PSPACE-hard in five different ways, corresponding to different subsets of the game mechanics.
In particular, we prove the game PSPACE-hard even without player input.
\end{abstract}

\section{Introduction}

\defn{Celeste}\footnote{\url{https://exok.com/games/celeste/}.  Celeste and its sprites are the properties of Maddy Makes Games.  Sprites are used here under Fair Use for the educational purpose of illustrating mathematical theorems.}
is a 2D platform video game released in 2018 by Maddy Makes Games.
It won the Best Independent Game and Games for Impact awards at The Game Awards 2018 \cite{GameAwards} and sold over a million copies \cite{EXOKProfile}.
In Celeste, the player controls a single character, Madeline, who must navigate various hazards along her journey.  We analyze the following natural decision problem about Celeste:

\begin{definition}[\Celeste]
  Given a Celeste level, is it possible for Madeline to traverse from a designated start location to a designated end location?
\end{definition}

A previous paper \cite{Ahmed2022Celeste} attempted to resolve this question by claiming that {\Celeste} is NP-complete, but failed to correctly prove containment in NP.
Specifically, their proof made the incorrect assumption that the sequence of inputs solving a {\Celeste} instance must be polynomially bounded in size.
Their NP-hardness reduction applies the framework from \cite{Nintendo_TCS} to show hardness with barriers, gates, and buttons.%
\footnote{Their clause gadget needs some modification to respect Celeste's mechanic that a button opens the Euclidean-nearest gate, but this is easy to do.}
They also showed that adding additional mechanics to Celeste (buttons that close gates instead of opening them) suffices for PSPACE-hardness, using the pressure-plate framework of \cite{Gaming_2014}.
By contrast, we show that Celeste's built-in mechanics, excluding gates and buttons, suffice for PSPACE-hardness.

We give five proofs that {\Celeste} is PSPACE-hard, each using different restricted combinations of existing game mechanics; see Table~\ref{tab:summary}.
Four of these proofs involve constructing a polynomial-time reduction to {\Celeste} from a motion-planning problem through a planar network of doors \cite{Doors_FUN2020}.
We make use of both ``open--close--traverse'' doors, as introduced in \cite{Gaming_2014, Lemmings_2015, Nintendo_TCS} and shown not to need crossovers in \cite{Doors_FUN2020}, and ``self-closing doors'', as introduced in \cite{Doors_FUN2020}.
In each case we construct a Celeste level corresponding to the motion-planning problem that can be traversed if and only if the motion-planning problem is solvable.
In all but one case we additionally show containment in PSPACE, establishing PSPACE-completeness.

\definecolor{hard}{rgb}{1,0.85,0.85}
\definecolor{open}{rgb}{1,1,0.85}
\definecolor{easy}{rgb}{0.85,0.85,1}
\definecolor{header}{rgb}{0.85,0.85,0.85}

\arrayrulewidth=0.75pt
\begin{table}[t]
  \centering
  \def\arraystretch{1.2}
  \def\FONT{\footnotesize}
  \def\ALLOW{\checkmark}
  \def\UNNEC{}
  \def\OPEN{\emph{OPEN}}
  \def\STACK#1#2{$\vcenter{\hbox{#1}\hbox{#2}}$}
  \def\IMAGE#1{$\vcenter{\hbox{\includegraphics{figures/entities/#1}}\vskip2pt}$}
  \tabcolsep=2.5pt
  \begin{tabular}{|cccccccccc||c|>{\footnotesize}c|}
    \hline
    \rowcolor{header}
    \multicolumn{1}{|c|}{\FONT \makecell{jump-\\through}} & \multicolumn{1}{c|}{\FONT \makecell{crumble\\blocks}} & \multicolumn{1}{c|}{\FONT spinners} & \multicolumn{1}{c|}{\FONT spring} & \multicolumn{1}{c|}{\FONT seeker} & \multicolumn{1}{c|}{\FONT jellyfish} & \multicolumn{1}{c|}{\FONT pufferfish} & \multicolumn{1}{c|}{\FONT barrier} & \multicolumn{1}{c|}{\FONT \makecell{move\\blocks}} & \multicolumn{1}{c||}{\FONT \makecell{Kevin\\blocks}} &&
    \\
    \rowcolor{header}
    \multicolumn{1}{|c|}{\IMAGE{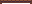}} & \multicolumn{1}{c|}{\IMAGE{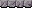}} & \multicolumn{1}{c|}{\IMAGE{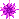}} & \multicolumn{1}{c|}{\IMAGE{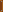}} & \multicolumn{1}{c|}{\IMAGE{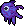}} & \multicolumn{1}{c|}{\IMAGE{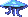}} & \multicolumn{1}{c|}{\IMAGE{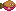}} & \multicolumn{1}{c|}{\IMAGE{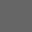}} & \multicolumn{1}{c|}{\IMAGE{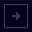}} & \multicolumn{1}{c||}{\IMAGE{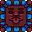}} & \multicolumn{1}{c|}{\raisebox{2.5ex}{Complexity}} & \multicolumn{1}{c|}{\normalsize\raisebox{2.5ex}{Sec}}
    \\
    \hline
    \hline
    \rowcolor{hard}
    \UNNEC & \UNNEC & \ALLOW && \ALLOW &&& \ALLOW & \ALLOW &&
    PSPACE-hard & \S\ref{sec:seeker-barrier-move}
    \\
    \rowcolor{hard}
    \UNNEC & \UNNEC & \ALLOW &&& \ALLOW && \ALLOW &&&
    PSPACE-hard & \S\ref{sec:jellyfish-barrier}
    \\
    \rowcolor{hard}
    \UNNEC & \UNNEC & \ALLOW &&&& \ALLOW &&&&
    PSPACE-hard & \S\ref{sec:pufferfish}
    \\
    \rowcolor{hard}
    \UNNEC & \UNNEC & \ALLOW &&&&&&& \ALLOW &
    PSPACE-hard & \S\ref{sec:kevin}
    \\
    \rowcolor{hard}
    &&& \ALLOW && \ALLOW &&& \ALLOW &&
    PSPACE-hard & \S\ref{sec:zeroplayer}
    \\
    \hline
    \rowcolor{easy}
    \ALLOW & \ALLOW & \ALLOW & \ALLOW & \ALLOW & \ALLOW & \ALLOW & \ALLOW & \ALLOW &&
    $\in$ PSPACE & \S\ref{sec:in-pspace}
    \\
    \hline
  \end{tabular}
  \caption{Summary of our results about the complexity of {\Celeste}.}
  \label{tab:summary}
\end{table}

We also consider the following problem, where we ignore the player and treat Celeste as an automaton:

\begin{definition}[\ZeroCeleste]
  Given a Celeste level and Madeline's starting position, will Madeline ever reach a designated end position if the player makes no inputs?
\end{definition}

One of our proofs (Section~\ref{sec:zeroplayer}) shows that {\ZeroCeleste}
with a certain combination of game mechanics
(springs, jellyfish, and move blocks) is PSPACE-complete
by reducing from a zero-player motion-planning problem \cite{trains2020}.
Adapting this result gives our fifth proof that {\Celeste} is PSPACE-hard.

We have tested our constructions in the Celeste game itself;
the custom map file and
some dynamic illustrations can be found on an accompanying website.%
\footnote{\url{https://github.com/cryslith/celeste-constructions}}

\section{Definitions}
\label{sec:defs}

We define an idealized version of Celeste which captures its relevant behavior.
A Celeste level\footnote{Chapters in Celeste consist of many levels, each of whose state does not persist across level boundaries.  There is no bound on the size of a level; our reductions will each produce a single large level rather than a collection of small ones.} of size \(M\) consists of a subset of \(\mathbb{Z}^2\) called the \defn{tilemap}, which defines where the solid walls are in the level.
It additionally contains a polynomially sized set of entities of various types, each placed at a specific initial location.  The locations of entities are not necessarily aligned with the integer tile grid, but we require the tilemap and entities to be confined to an \(M \times M\) rectangle.  In Celeste, many quantities such as positions, sizes, and velocities, are stored as either 32-bit integers or 32-bit floating point numbers.  We will ignore unusual behavior caused by overflow or loss of precision, and instead work with the assumption that Celeste physics are translation-invariant.

The player interacts with Celeste by controlling the main character, Madeline; see Figure~\ref{fig:Madeline}.  Madeline can move left and right on the ground or in the air, jump off of the ground or walls, and climb walls for a limited time.  She can also \defn{dash} in any of the eight cardinal and ordinal directions, which gives her a short-lasting boost of speed in that direction.  She can dash only once in the air, after which she has to land on the ground before dashing again.%
\footnote{Later in the game, Madeline gains the ability to dash twice before landing on the ground, but we ignore this additional ability as it does not affect our PSPACE containment results. Our hardness constructions could also be modified to be robust to a constant number of dashes.}

Madeline can also perform a number of more advanced movements, such as superdashes, hyperdashes, ultras, wallbounces, and crouch-dashes.  However, these do not allow her to traverse our gadgets in unexpected ways, and are not necessary for their intended operation.  For instance, superdashes, hyperdashes, and ultras can impart more horizontal velocity than usual when jumping from the ground, but breaking most of our gadgets would require Madeline to gain vertical height rather than horizontal distance.  Wallbounces can give additional vertical height but require a suitably-placed wall, which we avoid.  Crouch-dashes can sometimes be used to bypass obstacles by squeezing between them, but we place our obstacles close enough together to prevent this.  It is possible that an undiscovered bug or exploit in the game physics could break our gadgets, but in the following we will assume this is not the case.

\begin{figure}
  \centering
  \includegraphics[width=1cm]{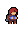}\qquad
  \includegraphics[width=4cm]{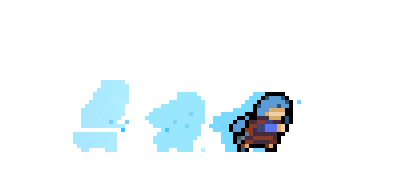}
  \caption{Madeline, at rest and while dashing}
  \label{fig:Madeline}
\end{figure}

Ordinarily, while falling Madeline can achieve a maximum ``glide ratio'' of at most 0.563.  That is, for every tile she falls vertically, she can drift only 0.563 tiles horizontally\footnote{Detailed information on Celeste's mechanics was obtained from many sources, including reverse-engineering the game code.}.  Dashing horizontally can move her up to 9.6 tiles horizontally, while dashing diagonally upwards can move her by at most 8 tiles horizontally and 2.1 tiles upward, after which her glide ratio decays to its baseline.  We exploit these limits to require Madeline to expend dashes and utilize entities in order to traverse sections of our constructions.

We make use of the following entities in our constructions.

\begin{itemize}[leftmargin=0.6in]
\def\ENTITYDEFAULT{1.5ex}
\let\ENTITYRAISE=\ENTITYDEFAULT
\def\ENTITY#1{\item[\smash{\raisebox{\ENTITYRAISE}{\raisebox{-\height}{\hbox to 0.6in{\hss\includegraphics{figures/entities/#1}\hss}}}}]}
\def\ENTITIES#1#2{\item[\smash{\raisebox{\ENTITYRAISE}{\raisebox{-\height}{\hbox to 0.6in{\hss\vbox{\hbox{\includegraphics{figures/entities/#1}}\vskip0.5ex\hbox{\includegraphics{figures/entities/#2}}}\hss}}}}]}
\def\ENTITYRAISE{1ex}\ENTITIES{jumpthrough}{crumble}\let\ENTITYRAISE=\ENTITYDEFAULT
  \defn{Jumpthroughs} can be passed through from below, but not from above.
  Madeline activates \defn{crumble blocks} by contacting them from above or from the sides (but not from the bottom),
  which makes them briefly disappear.
  These blocks act as one-way ``diodes'' which Madeline can traverse only in one direction (upwards and downwards respectively).
\ENTITY{spinner}
  \defn{Spinners} instantly kill Madeline upon her colliding with them,
  respawning her at her starting position and resetting the entire level to its initial state.
  This is never beneficial for the player.
  Spinners do not interact with entities other than Madeline.
\ENTITY{spring}
  \defn{Springs} in the pictured vertical orientation
  launch Madeline and jellyfish sideways and slightly upwards.
  Horizontal springs launch Madeline and jellyfish directly upwards while preserving some of their existing horizontal momentum.
\ENTITY{seeker}
  \defn{Seekers} kill Madeline on contact.  They can be pushed around by moving blocks, but otherwise remain still.
  If they ever come into line-of-sight of Madeline, they chase her and exhibit more complicated behavior;
  we will avoid defining this behavior by forcing Madeline
  to either avoid line-of-sight or contact the seeker in our constructions.
\ENTITY{jelly}
  \defn{Jellyfish} can be held by Madeline to float through the air when falling, increasing her horizontal speed and reducing her terminal velocity.
  This improves her maximum glide ratio to 4.5.
  Madeline can also throw jellyfish a short distance.
\ENTITY{puffer}
  Madeline can interact with \defn{pufferfish} by jumping on top of them, which restores her dash and moves the pufferfish downwards a short distance.
  If she instead approaches the pufferfish from below or to the side, the pufferfish will explode, restoring her dash and bouncing her away.
  The pufferfish will then respawn at its original position.
\def\ENTITYRAISE{2.5ex}\ENTITY{barrier}\let\ENTITYRAISE=\ENTITYDEFAULT
  \defn{Barriers} are intangible to Madeline, but act as solid walls for seekers.  They also permanently destroy any jellyfish that contact them.
\ENTITY{move}
  Madeline activates a \defn{move block} by touching it from above or from the sides,
  which causes it to begin moving in the direction displayed as an arrow on the block.
  The block will keep moving until obstructed by a solid wall, at which point it disappears and respawns at its original position.
  A move block with a spring attached will also be activated if Madeline or a jellyfish contacts the spring.
\ENTITY{kevin}
  \defn{Kevin blocks}\footnote{Named after Kevin Regamey, Celeste's sound designer.} have complex behavior.
  Each Kevin block maintains a stack of (position, direction) pairs, initially empty.
  Madeline activates the Kevin block by dashing into it from any side.
  Upon activation, the Kevin block pushes its current position onto its stack along with the direction Madeline dashed into it from,
  \emph{provided} that the direction currently on top of the stack (if any) is perpendicular to this new direction.
  The effect of this condition is that consecutive triples of positions on the stack always form right angles;
  they are never collinear.
  In any case, the Kevin block then begins charging quickly in the direction \emph{toward} the side that Madeline hit, until the Kevin block hits a wall.
  Whenever it is not charging, the Kevin block will retrace its path by slowly moving towards the position on top of its stack,
  removing that position from the stack when it arrives there.
\end{itemize}

\label{sec:in-pspace}
\begin{lemma}
  {\Celeste} with the listed entities other than Kevin blocks is contained in PSPACE.
\end{lemma}
\begin{proof}
  The states of Madeline and of every entity other than Kevin blocks can be described by a polynomial amount of space, since entities are confined to a polynomial-sized area.
  Thus an algorithm which guesses Madeline's inputs on every frame and simulates Celeste,
  until Madeline either reaches the goal location or a state is repeated,
  requires only polynomial space to function,
  showing containment in NPSPACE = PSPACE \cite{Savitch-1970}.
\end{proof}

\section{Single-Player Hardness}

In order to show PSPACE-hardness, we simulate certain ``door gadgets''.
We make use of both ``open--close--traverse'' doors, as introduced in \cite{Gaming_2014, Lemmings_2015, Nintendo_TCS}, and ``self-closing doors'', as introduced in \cite{Doors_FUN2020}.
Refer to Figure~\ref{fig:doors}.
An \defn{open--close--traverse door} is a 2-state gadget with three tunnels labeled ``open'', ``close'', and ``traverse''.  Traversing the open and close tunnels respectively changes the gadget's state to open or closed; the traverse tunnel can be traversed only in the open state.
A \defn{self-closing door} is a 2-state gadget with two tunnels labeled ``open'' and ``self-close''; traversing the self-close tunnel is possible only in the open state and changes the gadget's state to closed, preventing it from being traversed again until the open tunnel is traversed and the state is reset to open.
A \defn{symmetric self-closing door} is similar, except that the open tunnel (more symmetrically called a ``self-open'' tunnel) cannot be traversed while the door is in the open state.
All tunnels are \defn{directed}: they can be traversed in only one direction.
Optionally, the open--close--traverse door or self-closing door
(but not the symmetric self-closing door) can be made
\defn{open-optional} meaning that the two end locations of the open tunnel
are identified, making the open tunnel into an open \defn{port},
so the agent can always freely choose whether to open the
door or just skip the traversal.

\begin{figure}
  \centering
  \subcaptionbox{Open--close--traverse door}{\includegraphics[scale=0.8]{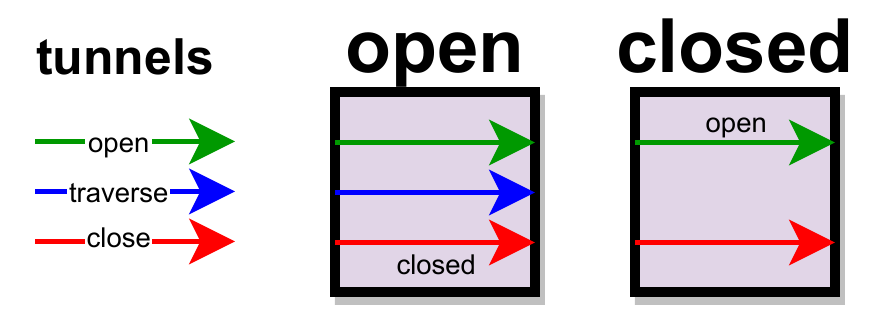}}\hfill
  \subcaptionbox{Self-closing door}{\includegraphics[scale=0.8]{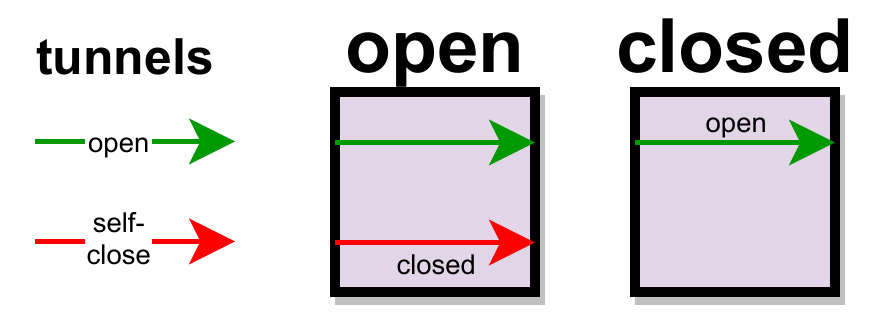}}

  \subcaptionbox{Symmetric self-closing door}{\includegraphics[scale=0.8]{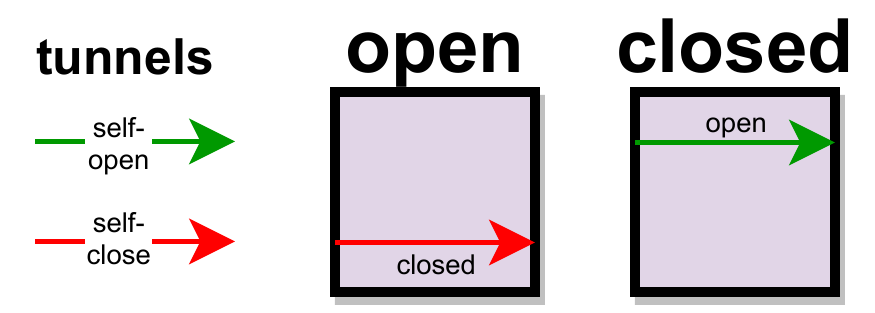}}
  \caption{The three types of door gadgets we use, from \cite{Doors_FUN2020}.
    Each diagram consists of a legend of labeled tunnels on the left,
    and the traversals possible in each of the two states of the gadget
    (``open'' and ``closed'') on the right.
    Each traversal that changes the state of the gadget
    is labeled with the state that it changes to.
  }
  \label{fig:doors}
\end{figure}

A key set of results from \cite{Doors_FUN2020} (Theorems 4.1 and 4.7) is that the ``planar motion planning problem'' is PSPACE-hard for any one of these types of doors\footnote{Except for a specific planar arrangement of the tunnels of an open--close--traverse door; we will avoid that arrangement.}.
More precisely, for any gadget, a \defn{planar network of gadgets}
consists of a finite number of instances of that gadget,
each with a specified initial state,
and an undirected graph connecting together the ends of tunnels
(called ``locations''), such that the gadgets and graph can be drawn
in the plane without crossings.
The \defn{planar (1-player) motion planning problem} asks, given a planar network
of gadgets, a start location, and a destination location, whether there is
a traversal sequence from start to destination.
Because it is easy to build ``hallways''
(paths that Madeline can traverse in any direction)
and ``branching hallways'' (connections where Madeline can freely choose to
follow any incident hallway), the graph part of a planar network is easy to
represent.
Therefore, constructing a gadget that simulates any one door is enough to show PSPACE-hardness of traversing Celeste levels.

The following four constructions all use jumpthroughs, crumble blocks, and
spinners, so we omit their mention in the section titles, and instead
just list the unique mechanics that each construction uses.
The jumpthroughs and crumble blocks are convenient shorthand for
one-way diodes, which can instead be replaced by the gadget in
Figure~\ref{fig:long-fall}, so they are not listed in the theorems or
Table~\ref{tab:summary}.

\begin{figure}
  \centering
  \includegraphics[scale=0.75]{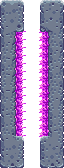}
  \caption{A one-way diode gadget, which can be traversed only from top to bottom.  Madeline cannot ascend from bottom to top even with a dash.}
  \label{fig:long-fall}
\end{figure}

\subsection{Seekers, Barriers, and Move Blocks}
\label{sec:seeker-barrier-move}

\begin{theorem}
  {\Celeste} with spinners, seekers, barriers, and move blocks is PSPACE-complete.
\end{theorem}
\begin{proof}
  We reduce from planar motion planning with open--close--traverse doors.
  Figure~\ref{fig:seeker} shows the simulation of the door.

\begin{figure}
  \centering
  \vspace*{1em}
  \begin{overpic}[scale=0.75]{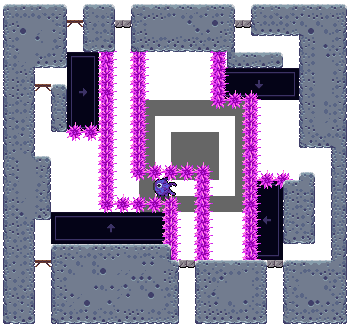}
    \put(18,94){\makebox(0,0){\strut open out}}
    \put(39,94){\makebox(0,0){\strut traverse in}}
    \put(69,94){\makebox(0,0){\strut close in}}
    \put(12,-2){\makebox(0,0){\strut open in}}
    \put(52,-2){\makebox(0,0){\strut traverse out}}
    \put(78,-2){\makebox(0,0){\strut close out}}
  \end{overpic}
  \vspace*{1em}
  \caption{An open--close--traverse door constructed with a seeker, barriers, and move blocks.  Currently in the ``closed'' state.}
  \label{fig:seeker}
\end{figure}
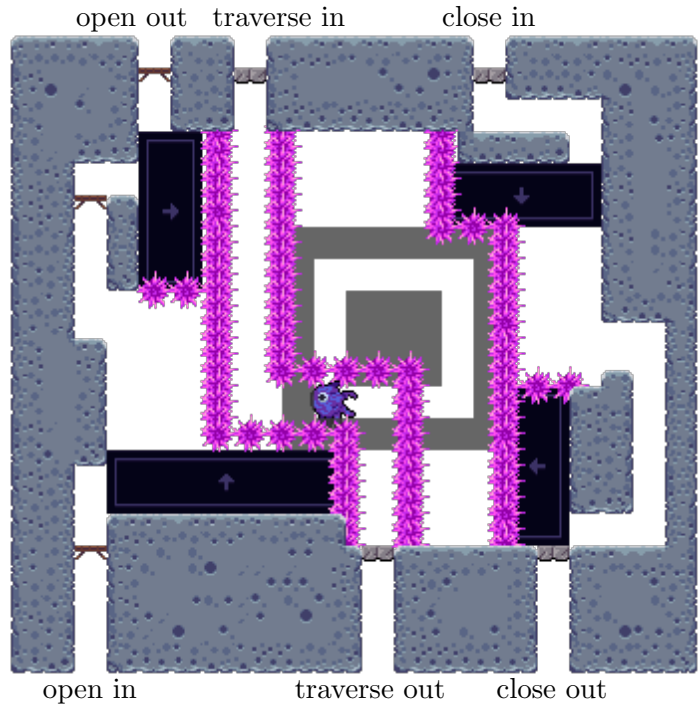

  The seeker is constrained by a ring of barriers.
  Whenever the seeker is in the top-right corner of the ring, Madeline can traverse the middle tunnel by falling down and dashing rightwards
  to avoid the spinners.
  However, when the seeker is in the bottom-left corner, she cannot traverse the middle tunnel without hitting it.
  The door is opened by traversing the left tunnel, which activates two move blocks, pushing the seeker to the top-right corner.
  (In particular, after activating the first move block, Madeline can follow
  close behind and get past it when the move block hits the solid wall,
  disappears, and respawns at its original location.)
  Similarly, traversing the right tunnel activates two other move blocks which return the seeker to the bottom-left corner.
  Jumpthroughs and crumble blocks ensure that the gadget does not reach an invalid state.
  For example, once Madeline has triggered one of the move blocks, she cannot
  exit the gadget (e.g., via the entrance she just used) except via the intended exit.
  Importantly, it is not possible for Madeline to be in line-of-sight of the seeker without dying;
  this is necessary to avoid triggering the seeker's complex chasing behavior.
\end{proof}

\subsection{Jellyfish and Barriers}
\label{sec:jellyfish-barrier}

\begin{theorem}
  {\Celeste} with spinners, jellyfish, and barriers is PSPACE-complete.
\end{theorem}
\begin{proof}
  We reduce from planar motion planning with open-optional self-closing doors.
  Figure~\ref{fig:jelly} shows the simulation of the door.

\begin{figure}
  \centering
  \vspace*{1em}
  \begin{overpic}[scale=0.75]{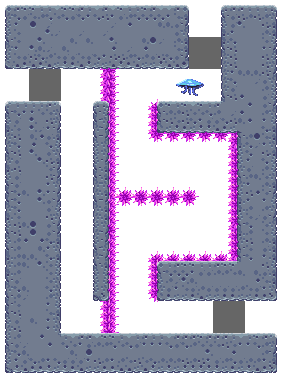}
    \put(0.5,77){\makebox(0,0)[r]{\strut open in/out}}
    \put(54,100.5){\makebox(0,0){\strut self-close in}}
    \put(74.5,15.5){\makebox(0,0)[l]{\strut self-close out}}
  \end{overpic}
  \caption{An open-optional self-closing door using a jellyfish and barriers, in the ``open'' state.}
  \label{fig:jelly}
\end{figure}

  The right-side self-closing tunnel can be traversed from top to bottom only if Madeline has a jellyfish,
  which she can use to drift around the spinners.
  However, she cannot traverse it without a jellyfish because her ordinary glide ratio is not enough to get around the corners;
  the best she can do is use her single dash to get around one of them.

  After Madeline traverses the right-side tunnel, the jellyfish is stuck at the bottom with her.
  The only way to return it to the top of the gadget is to first throw it through the spinners into the left chamber.
  The door is reopened by entering the left chamber, retrieving the jellyfish from the bottom,
  and throwing it back across the spinners at the top.
  The jellyfish can never exit the gadget because of the barriers blocking the entrances.
  (Recall that barriers permanently destroy jellyfish, so they also cannot be
  used to reset the jellyfish's position and re-open the gadget.)
\end{proof}

\subsection{Pufferfish}
\label{sec:pufferfish}

\begin{theorem}
  {\Celeste} with spinners and pufferfish is PSPACE-complete.
\end{theorem}
\begin{proof}
  We reduce from planar motion planning with symmetric self-closing doors.
  Figure~\ref{fig:puffer} shows the simulation of the door.

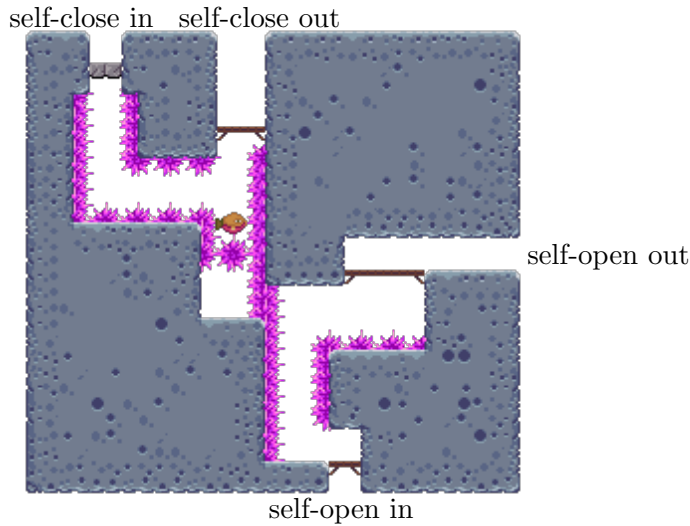
\begin{figure}
  \centering
  \vspace*{1em}
  \begin{overpic}[scale=0.75]{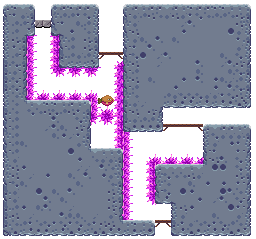}
    \put(12,94){\makebox(0,0){\strut self-close in}}
    \put(47,94){\makebox(0,0){\strut self-close out}}
    \put(99.5,47.5){\makebox(0,0)[l]{\strut self-open out}}
    \put(63,-2.5){\makebox(0,0){\strut self-open in}}
  \end{overpic}
  \vspace*{1em}
  \caption{A symmetric self-closing door using a pufferfish.  Initially the top tunnel is open and the bottom tunnel closed.}
  \label{fig:puffer}
\end{figure}

  Initially the bottom tunnel is untraversable, and Madeline can traverse the top tunnel only by
  dropping in, dashing right, jumping on the pufferfish, and dashing upwards to exit.
  This moves the pufferfish downwards through the spinners into the central area.

  Now the top tunnel cannot be traversed (because Madeline has only a single dash),
  but Madeline can traverse the bottom tunnel by dashing upwards into the pufferfish's explosion radius.  The pufferfish's explosion launches her to the right wall, which she can climb to the exit.
  The pufferfish then respawns in its original position, resetting the gadget.

  A minor issue with this construction is that because the pufferfish always begins the level at its spawn point,
  the gadget cannot be initialized with the bottom tunnel open and the top tunnel closed.
  This is solved by reflecting the gadget about the $y$ axis to obtain a door gadget that initially has the opposite tunnel open.
\end{proof}

\subsection{Kevin Blocks}
\label{sec:kevin}

In order to show that {\Celeste} is PSPACE-hard with Kevin blocks, we will construct an open-optional self-closing door gadget.
One difficulty is that our construction always begins in the closed state,
but we will show that this nonetheless suffices for PSPACE-hardness.

\begin{lemma} \label{lem:initially closed}
  Planar motion planning with initially closed open-optional self-closing doors
  is PSPACE-complete.
\end{lemma}
\begin{proof}
  We reduce from motion planning with open-optional self-closing doors
  that may begin in either state.

  First, we use a combination of open-optional self-closing doors to create an open-optional self-closing door with two opening ports,\footnote{The same construction is used for a different purpose in Theorem 3.3 of \cite{Doors_FUN2020}.} as shown in Figure~\ref{fig:initiallyclosed}.
  This construction needs one-way diodes, which are easy to implement with a
  directed open-optional self-closing door by identifying the open port
  with the entrance to the self-close tunnel.
  We replace every initially open door in the given instance with one of these doors, giving each one an extra opening port.
  Finally, we build a new traversal path for the agent that
  starts from a new start location,
  visits the extra opening port of each initially open door in sequence,
  and then proceeds to the original start location of the given instance.
  We place a one-way diode at the end of the path, so the agent can visit all
  the extra ports at the beginning but never again.
  Opening doors to gadgets can only help later traversal, so we can assume
  that the agent visits all the extra ports, and thus opens all doors that
  were supposed to be initially open.

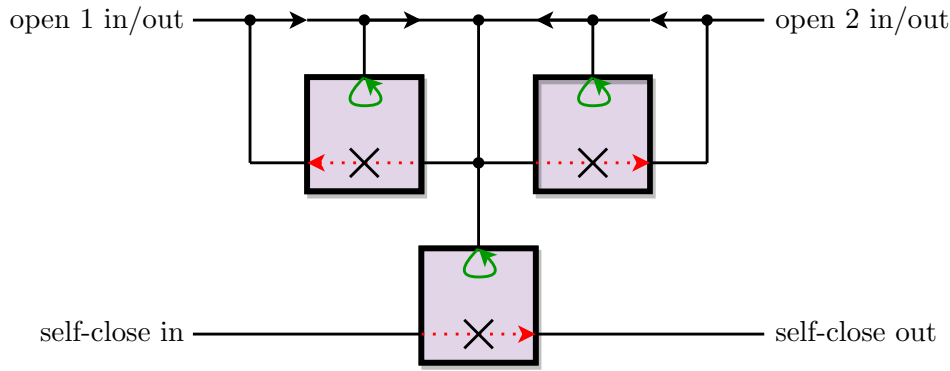
\begin{figure}
  \centering
  \begin{overpic}[scale=0.75]{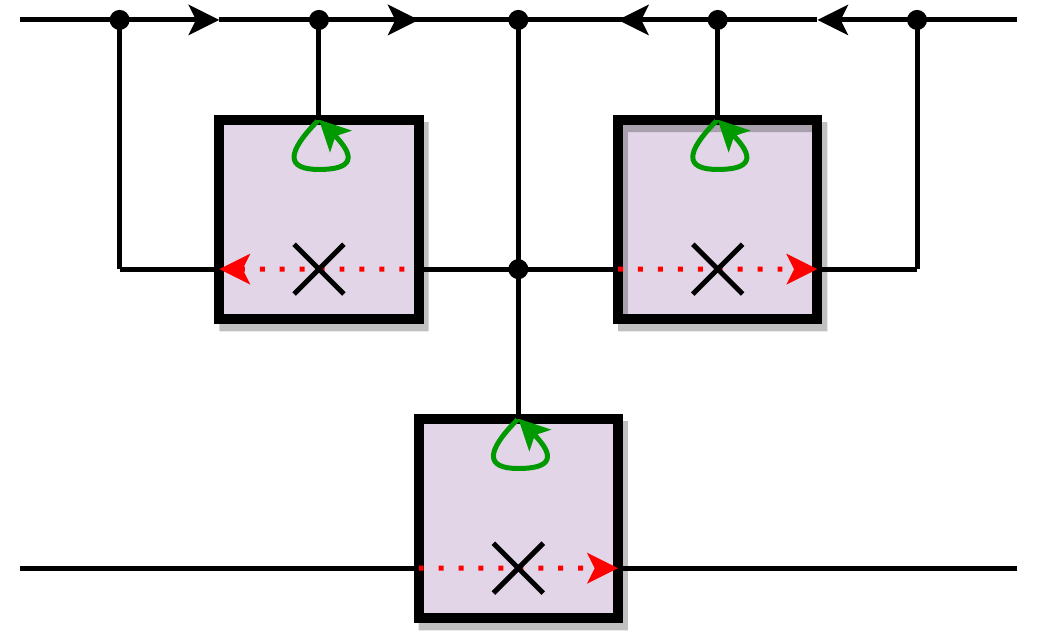}
    \put(0.5,59){\makebox(0,0)[r]{\strut open 1 in/out}}
    \put(99.5,59){\makebox(0,0)[l]{\strut open 2 in/out}}
    \put(0.5,6.5){\makebox(0,0)[r]{\strut self-close in}}
    \put(99.5,6.5){\makebox(0,0)[l]{\strut self-close out}}
  \end{overpic}
  \caption{Duplicating the opening port of a directed
    open-optional self-closing door:
    both of the two upper ports open the bottom self-close tunnel,
    without it being possible to leak between the two ports.
    Each box denotes an open-optional self-closing door,
    where the green loop is the open port and the
    dotted arrow denotes the self-close tunnel in an initially closed state.
    The arrows at the top of the diagram (exterior to gadgets)
    are one-way diodes.
  }
  \label{fig:initiallyclosed}
\end{figure}

  This construction does not preserve planarity.
  Luckily, the simulation of a directed crossover in Theorem 4.1 of \cite{Doors_FUN2020} uses only initially closed doors.
  Thus we can make our construction planar by replacing any crossing wires with these simulated crossovers, while preserving that all doors are initially closed.
\end{proof}

Now we can show the following result about Kevin blocks.

\begin{theorem} \label{thm:kevin}
  {\Celeste} with spinners and Kevin blocks is PSPACE-hard.
\end{theorem}
\begin{proof}
  We reduce from planar motion planning with initially closed open-optional self-closing doors, which is PSPACE-complete according to Lemma~\ref{lem:initially closed}.
  Figure~\ref{fig:kevin} shows the simulation of the door.

\begin{figure}
  \centering
  \vspace*{1em}
  \begin{overpic}[scale=0.75]{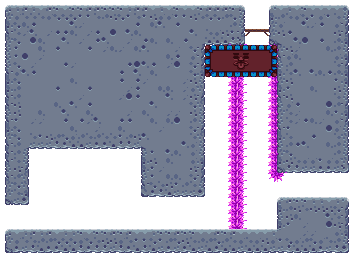}
    \put(0.5,10.5){\makebox(0,0)[r]{\strut open in/out}}
    \put(72,72.5){\makebox(0,0){\strut self-close out}}
    \put(99,19.5){\makebox(0,0)[l]{\strut self-close in}}
  \end{overpic}
  \caption{An open-optional self-closing door using a Kevin block.  Initially in the ``closed'' state.}
  \label{fig:kevin}
\end{figure}

  Initially the right tunnel is untraversable, since the Kevin block is placed too high up for Madeline to dash into it.
  By entering from the left, Madeline can climb the wall and
  move the Kevin block into the left chamber.
  By dashing into it repeatedly from different directions, she can ``wind up'' the Kevin block, adding arbitrarily many positions to its internal stack (e.g., by moving it clockwise around a rectangle).
  Once this is done, the block will take arbitrarily long to unwind, depending on how long it was wound up for.
  The right tunnel can be traversed only when the Kevin block finishes unwinding and returns to its original position,
  when Madeline can ride on top of it.
  The constrained space prevents the Kevin block from possibly moving through the right tunnel more than once per winding session,
  because positions can be added to its stack only when it makes right-angle turns.
  Therefore, every traversal of the right tunnel corresponds to an immediately previous visit to the left chamber,
  which is the same condition as a self-closing door.

  It follows that any traversal of a network of Kevin-based doors corresponds to a traversal of the corresponding network of self-closing doors.
  The only obstacle to showing the converse is in solving certain timing constraints:
  Whenever Madeline needs to go through an open door,
  the corresponding Kevin block must still be unwinding; that is, it must not have been wound for too short a time.
  Fortunately, these constraints are always solvable by the following method.
  List the doors in reverse order by the time at which the open tunnel is visited (each door will in general appear multiple times on the list).
  In this order, assign an amount of time to wind each door sufficient to keep the door wound until its self-close tunnel must be traversed.
  This quantity of time depends only on how much time was assigned for winding previous doors in the list.
  Because the constraints are always solvable, this shows that any traversal of the initially closed self-closing door network can be transformed into a traversal of the network of Kevin-based doors.
\end{proof}

\section{Zero-Player Hardness}
\label{sec:zeroplayer}

In this section we prove that predicting the outcome of a Celeste level is PSPACE-complete, ignoring player inputs.
We do so by reducing from a zero-player motion-planning problem,
in which an agent traverses a network of gadgets entirely deterministically.
Specifically, in \defn{zero-player motion planning} \cite{trains2020},
the gadgets must be \defn{input/output}, meaning that their locations can be
partitioned into entrances (inputs) and exits (outputs) ---
no location is an entrance for a transition in some state and
an exit for another transition in some state ---
and the connection graph connecting gadget locations must be
\defn{branchless}, meaning that it has at most one input location
in each connected component.
Thus the motion of the agent is fully determined from its starting location
and the initial gadget states; the goal is to determine whether the agent
ever reaches a given destination location.

The input/output gadget we consider is the
\defn{set-up switch/set-down switch}; refer to Figure~\ref{fig:switches}.
This 2-state gadget has two input locations,
labeled ``set-up'' and ``set-down'',
and four output locations,
labelled ``(up, up)'', ``(up, down)'', ``(down, up)'', and ``(down, down)''.
If an agent enters at input location set-$i$,
then they are forced to exit at output location $(i,s)$
where $s$ is the state of the gadget before traversal,
and then the gadget's state is set to~$i$.

\begin{figure}
  \centering
  \includegraphics[scale=0.8]{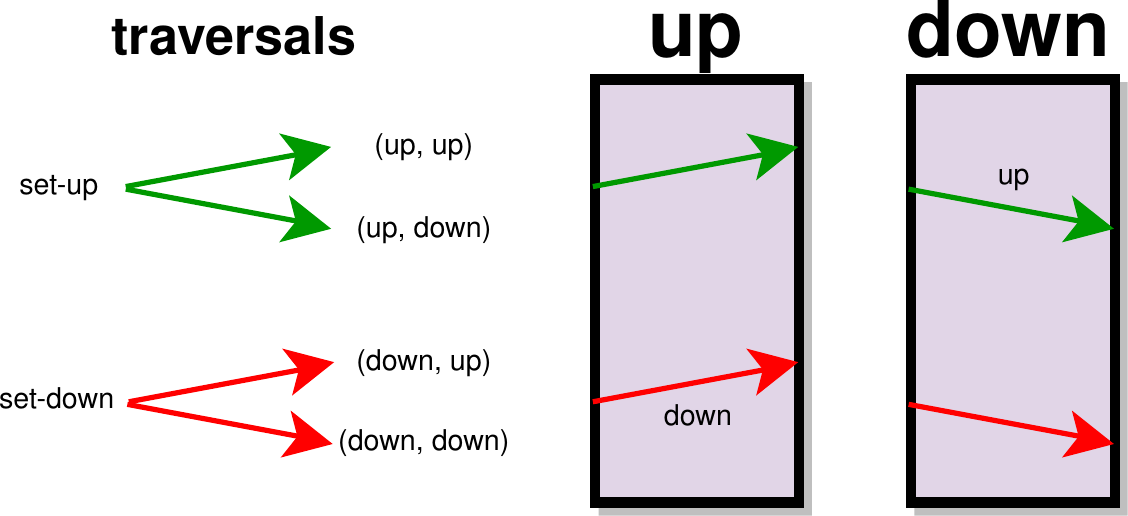}
  \caption{The set-up switch/set-down switch gadget.  Left: The six locations (two input and two output) along with arrows for all possible traversals between them.  Right: The traversals possible from each of the two states, with labels for traversals which modify the resulting state.}
  \label{fig:switches}
\end{figure}

Ani et al.~\cite{trains2020} gave a partial classification of
zero-player gadgets and the resulting complexity of the zero-player
motion planning problem.
An important result (Corollary 2.8) is that motion planning is PSPACE-complete
for any unbounded output-disjoint deterministic 2-state input/output gadget
with multiple nontrivial inputs.
Examining the set-up switch/set-down switch, we see that it is a 2-state
input/output gadget,
it is unbounded (can change states arbitrarily many times),
output-disjoint (no output can be reached from multiple inputs),
and deterministic (each input leads to a unique transition/output
in any given state).
Neither of the inputs is a trivial tunnel
(avoiding interacting with the state at all),
so the corollary shows that zero-player motion planning
with this gadget is PSPACE-complete.

We reduce this zero-player motion-planning problem to {\ZeroCeleste}
by simulating a set-up switch/set-down switch.
Unlike our other simulations, the motion-planning ``agent'' in this case
is not Madeline, but instead a jellyfish, so we do not need to worry about
player input interfering with the construction.
Also because of this, we do not need spinners, jumpthroughs, or crumble blocks
in this construction.

\begin{theorem}
  {\ZeroCeleste} with jellyfish, springs, and move blocks is PSPACE-complete.
\end{theorem}
\begin{proof}
  We reduce from zero-player motion planning with the set-up switch/set-down
  switch.  Figure~\ref{fig:switchsim} shows the simulation of this gadget.

  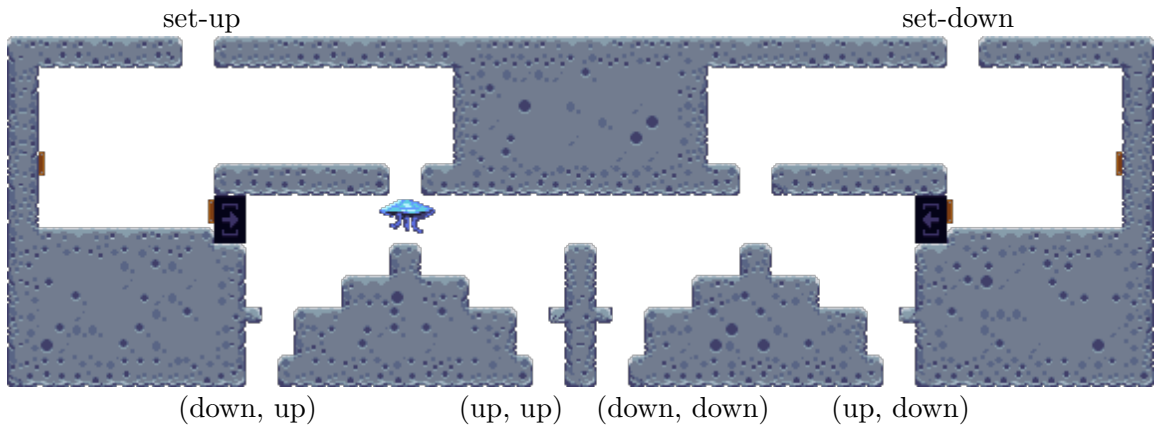
\begin{figure}
    \centering
    \vspace*{1em}
    \begin{overpic}[scale=0.75]{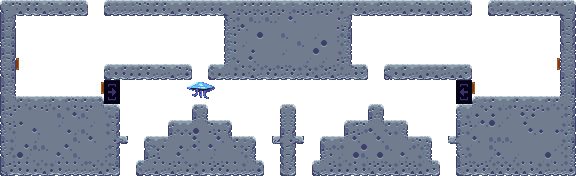}
      \put(17,32){\makebox(0,0){\strut set-up}}
      \put(83,32){\makebox(0,0){\strut set-down}}
      \put(21,-2){\makebox(0,0){\strut (down, up)}}
      \put(44,-2){\makebox(0,0){\strut (up, up)}}
      \put(59,-2){\makebox(0,0){\strut (down, down)}}
      \put(78,-2){\makebox(0,0){\strut (up, down)}}
    \end{overpic}
    \vspace*{1em}
    \caption{A set-up switch/set-down switch constructed with a jellyfish, springs, and move blocks.  Currently in the ``up'' state.}
    \label{fig:switchsim}
  \end{figure}

  The state of the gadget is determined by which side of the central chamber a jellyfish rests in.  A new jellyfish entering either of the input locations hits the spring on the side of the corresponding move block, activating it.  The move block pushes the original jellyfish out of the gadget through the corresponding output location.  Meanwhile the new jellyfish bounces off the spring on the outer wall and is propelled towards the middle of the gadget, where it falls through the hole into one of the central locations, setting the state of the gadget.

  It remains to show how to combine these gadgets together in a network.
  Figure~\ref{fig:hallways} shows how to route a jellyfish along
  predetermined paths to implement ``hallway'' connections between gadgets,
  including how to merge the output locations of multiple gadgets
  into a single input location of another gadget.
  In our construction, only one jellyfish at a time is outside of a gadget,
  and that jellyfish represents the agent
  (though \emph{which} jellyfish represents the agent varies over time).

  \begin{figure}
    \centering
    \vspace*{1em}
    \begin{subfigure}{0.9\textwidth}
      \centering
      \begin{overpic}[scale=0.8]{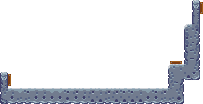}
        \put(0, 0){\includegraphics[scale=0.8]{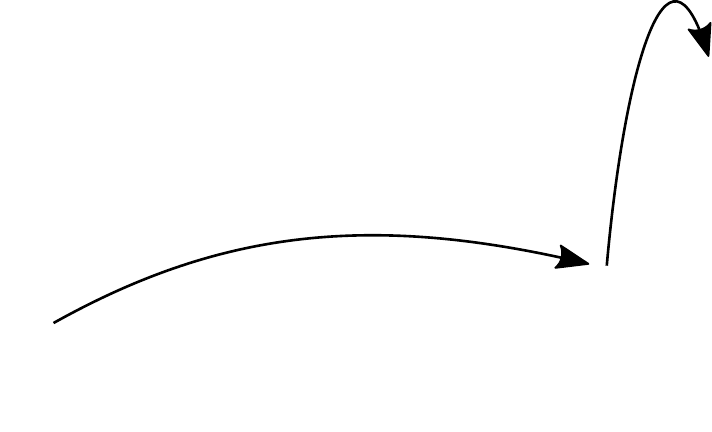}}
      \end{overpic}
      \caption{Modular hallway piece.  If the third spring is omitted, then
        jellyfish will fall straight down after traversing the hallway piece,
        allowing hallways to turn downwards or connect to gadgets.}
    \end{subfigure}

    \vspace*{2em}
    \subcaptionbox{A hallway that turns around, made from four hallway pieces}{
      \begin{overpic}[scale=0.8]{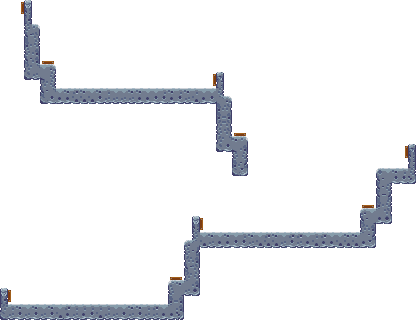}
        \put(0, 0){\includegraphics[scale=0.8]{figures/hallway_annotation}}
        \put(46, 17){\includegraphics[scale=0.8]{figures/hallway_annotation}}
        \put(51, 34){\scalebox{-1}[1]{\includegraphics[scale=0.8]{figures/hallway_annotation}}}
        \put(5, 52){\scalebox{-1}[1]{\includegraphics[scale=0.8]{figures/hallway_annotation}}}
      \end{overpic}
    }

    \vspace*{2em}
    \subcaptionbox{Hallways can merge together or cross over each other.}{
      \begin{overpic}[scale=0.8]{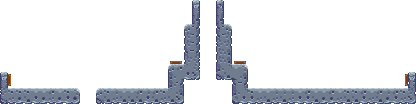}
        \put(0, 0){\includegraphics[scale=0.8]{figures/hallway_annotation}}
        \put(0, 0){\includegraphics[scale=0.8]{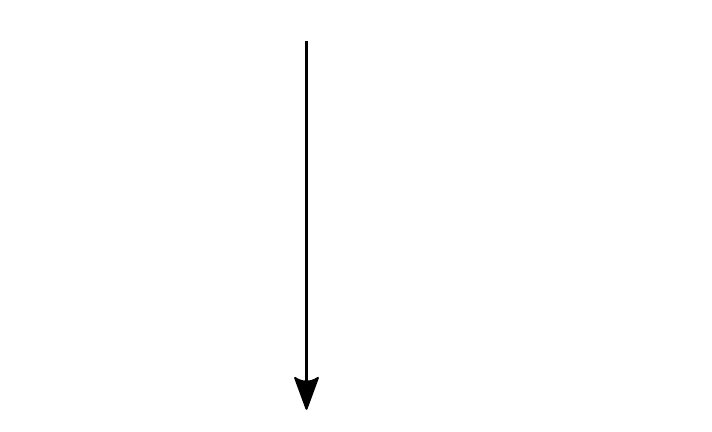}}
        \put(50.5, 0){\scalebox{-1}[1]{\includegraphics[scale=0.8]{figures/hallway_annotation}}}
      \end{overpic}
    }
    \caption{The various types of hallway connections needed to route jellyfish between gadgets.}
    \label{fig:hallways}
  \end{figure}

  We also need a final gadget at the destination location, so that
  Madeline reaches her destination without player inputs if and only if
  a jellyfish ends up at this gadget.
  Figure~\ref{fig:zeroplayerend} gives one such gadget.
\end{proof}

  \begin{figure}
    \centering
    \includegraphics[scale=1.0]{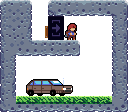}
    \caption{The final gadget.  Without player inputs, Madeline reaches her destination (by her car) if and only if a jellyfish enters this gadget.}
    \label{fig:zeroplayerend}
  \end{figure}

\begin{corollary}
  {\Celeste} with jellyfish, springs, and move blocks is PSPACE-complete.
\end{corollary}
\begin{proof}
  Using an alternate final gadget, we can instead trap Madeline beneath a move block until a jellyfish enters the gadget, as shown in Figure~\ref{fig:zerotoone}.
  (Recall that move blocks cannot be triggered from below.)
  \begin{figure}
    \centering
    \includegraphics[scale=1.0]{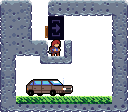}
    \caption{An alternate final gadget.  Madeline is able to reach her destination (with player inputs) if and only if a jellyfish enters this gadget.}
    \label{fig:zerotoone}
  \end{figure}
\end{proof}

\section{Conclusion}

We have shown that five variants of Celeste are PSPACE-hard to navigate
by reducing from motion-planning problems,
as summarized in Table~\ref{tab:summary}.
We conclude with a discussion of open questions.

\begin{question}
  What other subsets of Celeste mechanics suffice for PSPACE-hardness?
\end{question}

We can consider strict subsets of the mechanics already shown hard, or
consider other incomparable subsets, to fill in Table~\ref{tab:summary}.
Alternatively, we could consider other mechanics than those analyzed here;
see \cite{fandom-wiki,ink-wiki} for longer lists than Section~\ref{sec:defs}.

Of particular note is that Kevin blocks maintain a potentially unbounded stack of locations, which could store more than a polynomial amount of information.  However, it is difficult to make use of this information because they are always unwinding towards their initial states when not interacted with.
\begin{question}
  Is {\Celeste} with Kevin blocks contained in PSPACE?
\end{question}

In our reduction to {\Celeste} with Kevin blocks (Section~\ref{sec:kevin}), we constructed a self-closing door gadget which starts out closed.  We showed that motion planning over a planar network of initially closed self-closing doors is PSPACE-complete.

\begin{question}
  How does restricting various gadgets to particular initial states affect the complexity of the corresponding motion-planning problems?
\end{question}

\section*{Acknowledgments}

This work was initiated during extended problem solving sessions with the participants of the MIT class on
Algorithmic Lower Bounds: Fun with Hardness Proofs (6.892)
taught by Erik Demaine in Spring 2019.
We thank the other participants for their insights and contributions.
In particular, we thank Dylan Hendrickson and Josh Brunner for helping simplify the proof of Theorem~\ref{thm:kevin},
and we thank Michael Coulombe and Jayson Lynch for reading and suggesting revisions to drafts of the paper.

Images of sprites and gadgets were composed and tested using the fan-made level editor Ahorn\footnote{\url{https://github.com/CelestialCartographers/Ahorn}}.  We thank Cruor, Vexatos, and Ahorn's other contributors for creating this excellent tool.  We additionally thank the Celeste speedrunning, modding, and Tool-Assisted Speedrunning community for extensively researching Celeste's mechanics.

Finally, we thank Maddy Thorson, Noel Berry, and the rest of the development team for creating Celeste, a difficult and wonderful experience in many ways.

\bibliography{sources}

\newcommand{\etalchar}[1]{$^{#1}$}
\begin{thebibliography}{ADGV15}

\bibitem[ABD{\etalchar{+}}20]{Doors_FUN2020}
Joshua Ani, Jeffrey Bosboom, Erik~D. Demaine, Yevhenii Diomidov, Dylan
  Hendrickson, and Jayson Lynch.
\newblock Walking through doors is hard, even without staircases: Proving
  {PSPACE}-hardness via planar assemblies of door gadgets.
\newblock In {\em Proceedings of the 10th International Conference on Fun with
  Algorithms (FUN 2020)}, pages 3:1--3:23, La Maddalena, Italy, September 2020.

\bibitem[ACG{\etalchar{+}}22]{Ahmed2022Celeste}
Zeeshan Ahmed, Alapan Chaudhuri, Kunwar Grover, Ashwin Rao, Kushagra Garg, and
  Pulak Malhotra.
\newblock Classifying {C}eleste as {NP} complete.
\newblock In {\em Proceedings of the 9th International Conference on
  Foundations of Computer Science \& Technology}, Chennai, India, November
  2022.

\bibitem[ADGV15]{Nintendo_TCS}
Greg Aloupis, Erik~D. Demaine, Alan Guo, and Giovanni Viglietta.
\newblock Classic {N}intendo games are (computationally) hard.
\newblock {\em Theoretical Computer Science}, 586:135--160, 2015.

\bibitem[ADHL22]{trains2020}
Joshua Ani, Erik~D. Demaine, Dylan~H. Hendrickson, and Jayson Lynch.
\newblock Trains, games, and complexity: 0/1/2-player motion planning through
  input/output gadgets.
\newblock In {\em Proceedings of the 16th International Conference and
  Workshops on Algorithms and Computation (WALCOM 2022)}, Jember, Indonesia,
  2022.
\newblock arXiv:2005.03192.

\bibitem[{Cel}]{ink-wiki}
{Celeste Wiki}.
\newblock Mechanics.
\newblock \url{https://celeste.ink/wiki/Mechanics}.

\bibitem[{Fan}]{fandom-wiki}
{Fandom: Celeste Wiki}.
\newblock Objects.
\newblock \url{https://celestegame.fandom.com/wiki/Objects}.

\bibitem[Gra18]{GameAwards}
Christopher Grant.
\newblock {T}he {G}ame {A}wards 2018: Here are all of the winners.
\newblock {\em Polygon}, December 2018.

\bibitem[Mar20]{EXOKProfile}
Tom Marks.
\newblock Inside {EXOK} {G}ames: The brand new studio that's already sold a
  million copies.
\newblock {\em IGN}, March 2020.

\bibitem[Sav70]{Savitch-1970}
Walter~J. Savitch.
\newblock Relationships between nondeterministic and deterministic tape
  complexities.
\newblock {\em Journal of Computer and System Sciences}, 4(2):177--192, 1970.

\bibitem[Vig14]{Gaming_2014}
Giovanni Viglietta.
\newblock Gaming is a hard job, but someone has to do it!
\newblock {\em Theory of Computing Systems}, 54(4):595--621, 2014.

\bibitem[Vig15]{Lemmings_2015}
Giovanni Viglietta.
\newblock Lemmings is {PSPACE}-complete.
\newblock {\em Theoretical Computer Science}, 586:120--134, 2015.

\end{thebibliography}
\bibliographystyle{alpha}

\end{document}